\documentclass{ifacconf}
\pdfoutput=1
\usepackage{graphicx}         
\usepackage{amsmath} 
\usepackage{subfigure}
\usepackage{epstopdf}
\usepackage{color}
\usepackage{xcolor}
\usepackage{amssymb} 
\usepackage{amsfonts} 
\usepackage{algorithm}
\usepackage{algorithmic}
\usepackage{ulem}
\usepackage{url}
\usepackage{appendix}
\usepackage{lineno}
\usepackage{natbib} 
\usepackage{enumitem}

\newtheorem{theorem}{Theorem}

\newtheorem{proof}{Proof}
\newtheorem{lemma}{Lemma}

\newtheorem{assumption}{Assumption}

\begin{document}

\begin{frontmatter}

\title{Optimal Control Selection over the Edge-Cloud Continuum} 


\author[unipd]{Xiyu Gu\thanksref{corrauth},}
\author[unipd]{ Matthias Pezzutto,} 
\author[unipd]{ Luca Schenato,} 
\author[uppsala]{Subhrakanti Dey}

\address[unipd]{Department of Information Engineering, University of Padova, Italy}                                  
\address[uppsala]{Department of Electrical Engineering, Uppsala University, Sweden} 

\thanks[corrauth]{Corresponding author. E-mail: xiyu.gu@studenti.unipd.it}

\begin{abstract}
The emerging computing continuum paves the way for exploiting multiple computing devices, ranging from the edge to the cloud, to implement the control algorithm.
Different computing units over the continuum are characterized by different computational capabilities and communication latencies, thus resulting in different control performances and advocating for an effective trade-off. 
To this end, in this work, we first introduce a multi-tiered controller and we propose a simple network delay compensator. 
Then we propose a control selection policy to optimize the control cost taking into account the delay and the disturbances. We theoretically investigate the stability of the switching system resulting from the proposed control selection policy. Accurate simulations show the improvements of the considered setup. 
\end{abstract}


\end{frontmatter}

\section{Introduction}

Twenty years ago, the proliferation of shared high-performance computational units has started to move computing resources away from the user giving rise to the advent of cloud computing [\cite{velte2010cloud}]. In the opposite direction, in the last ten years, we have witnessed the return back of computational resources close to the user and the emergence of edge computing [\cite{shi2016edge}]. Indeed, although not able to implement as complex algorithms as the cloud, the edge can provide simpler solutions with a reduced delay to the user, thus providing another option for the computation offloading. Lately, a full spectrum of computing resources is appearing along the path from the cloud to the edge, often integrated into any network node. The resulting computation-communication ecosystem is known as computing continuum or cloud-edge continuum [\cite{al2024computing}]. To exploit its potential at best, computation offloading over the continuum needs to be efficiently orchestrated to trade-off between the computational power and the communication delay. This is critical for real-time decision making and control strategies since delays might not only deteriorate performances, but also jeopardize safety due to the underlying dynamical systems.

In the simplest case, a single computational unit from the computing continuum can be used to implement the control algorithm. This can be either the cloud, the edge on an arbitrary remote computing device as in the theory of NCS [\cite{pezzutto2024wireless}]. In this case, the main focus is on the robustness of the controller with respect to the communication delays. 
In a more general setup, the controller can be implemented over two devices. The work [\cite{liang2018control}] considers a local and a remote LQR controller, the works [\cite{pezzutto2021remote}][\cite{umsonst2024remote}] consider a local LQR and a remote MPC, while the work [\cite{skarin2020cloudhorizon}] considers a local and a remote MPC. In aforementioned works, the remote controller is applied as soon as it is available. 
The work [\cite{skarin2020cloudrate}] introduces a smooth transition between a remote MPC and a local LQR. In the works [\cite{ma2020exploring}][\cite{ma2022smart}], a selection logic based on the network conditions is used to choose between a remote MPC and an on-board LQR. In [\cite{li2022cloud}], a switching policy is used to swap from an initial open-loop sequence obtained by a cloud MPC to the closed-loop input from an edge MPC. 

In this work we consider the case where more than two computational resources are used. Ideally, we consider a controller on board, a controller on the edge device, and a controller on the cloud. We consider arbitrary control design with the underlying assumption that control accuracy and delay increase going further from the plant.
A suitable mechanism is introduced to compensate communication and computation delays. 
In order to effectively orchestrate the different computational devices, we introduce a novel control selection strategy which selects the best sequence taking into account the delay and the disturbances. We theoretically investigate the stability of the multi-tiered controller. Roughly speaking, we show that the plant is stable as long as the onboard controller is stabilizing and has ideal communications. Results can be generalized to any number of remote controllers over the computing continuum. Numerical simulations show the performance improvement compared to standard solutions.

\section{Problem formulation}

\subsection{Computing Continuum model}

We consider a three-tiered computation-communication architecture (Fig.~\ref{framework}) consisting of three computing units: a simple \textit{on-board} device physically located on the plant, an \textit{edge} computing device with good computational capabilities located in the same area of the plant, and a \textit{cloud} computing device with high computational capabilities located on a remote site.

We assume that the on-board device is connected to the sensors and actuators through ideal links. Conversely, the plant is connected to the edge through a wireless network and to the cloud through an internet connection provided by a heterogeneous communication system. Accordingly, communication latencies increase as the computing device is farther from the plant.

We assume that the on-board device can implement only simple algorithms (e.g. PID or LQR). Leveraging on their computational resources, we assume that the edge can implement more advanced algorithms (e.g. linear MPC) than the on-board device and that the cloud can implement more advanced algorithms than the edge (e.g. complex optimization algorithms). Ideally, we consider that control performance improves as more complex controller is used. 


﻿ 

﻿ 

﻿ 

﻿ 

\subsection{System model}
Consider a discrete-time nonlinear system
\begin{equation} \label{nonlinear model}
x(k+1) = f(x(k),u(k),w(k)),
\end{equation}
where $x(k) \in \mathbb{R}^n$ is the state, $u(k) \in \mathbb{R}^m$ is the input, $w(k) \in \mathbb{R}^n$ is the disturbance, and $f: \mathbb{R}^n \times \mathbb{R}^m \times \mathbb{R}^n \to \mathbb{R}^n$ is a nonlinear function. 

We assume that each computing device implements a predefined predictive controller with prediction horizon $N \in \mathbb{N}$. The control sequences provided by  cloud,  edge, and  on-board controller are respectively denoted as
\begin{align}
U_c(x) = \{u_0^c(x), \dots, u_{N-1}^c(x)\} = g_c(x), \label{eq:cloud_ctrl} \\
U_e(x) = \{u_0^e(x), \dots, u_{N-1}^e(x)\} = g_e(x), \label{eq:edge_ctrl} \\ 
U_b(x) = \{u_0^b(x), \dots, u_{N-1}^b(x)\} = g_b(x). \label{eq:board_ctrl} 
\end{align}

The plant is equipped with a buffer of length $N$ to store a backup control sequence at each sampling time $k$
\begin{equation}
B(k)=\{b_0(k),\dots, b_{N-1}(k)\}.
\end{equation}
We stress that each controller is standalone and the three devices provide three different alternatives of the input that can be directly applied by the actuator.

\subsection{Network model}

During each sampling period, the plant communicates with the edge device and the cloud device.  Let $\mathcal{X}_c(k)=\{x(k),B(k)\}$ be the packet transmitted by the plant to the cloud device at time $k$ and
$\mathcal{X}_e(k)=\{x(k),B(k)\}$ be the packet transmitted by the plant to the edge device at time $k$. We referred to them as uplink transmissions. Moreover, let
$\mathcal{U}_c(k)$ denote the packet transmitted by the cloud device to the plant at time $k$ and
$\mathcal{U}_e(k)$ denote the packet transmitted by the edge device to the plant at time $k$. We referred to them as downlink transmissions. The content of $\mathcal{U}_c(k)$ and $\mathcal{U}_e(k)$ is defined later.

The communications are not ideal but affected by random delays and packet losses. We denote as $d(\mathcal{X})$ the communication delay of a packet $\mathcal{X}$.
We denote as $d_c(k)$ and $d_e(k)$ the delay incurred in the transmission of $\mathcal{X}_c(k)$ and $\mathcal{X}_e(k)$
Moreover we introduce the age of information $a_c(k)$ at the cloud at time $k$ as the number of steps since the generation of the most recent packet received at the cloud. Formally 
$a_c(k)= \min \{a \ \text{s.t.} \ d(\mathcal{X}_c(k-a))\leq a \} $. It follows that $\mathcal{X}_c(k-a_c(k))$ is the most recent packet at time $k$ received by the cloud. The age of information $a_e(k)$ at the edge is defined similarly. 


We assume that each device is able to execute the corresponding controller in a negligible amount of time and we focus on the effect of the communication channels. Nevertheless, the running time of the controllers can be included in the communication delay.

\subsection{Objective}
The objective is to design a control selection strategy to select the input $u(k)$ to apply among the inputs provided by the three computing units in order to optimize the control cost
$\mathcal{V}(x(0),U) = \lim\nolimits_{N\rightarrow \infty } \frac{1}{N}\sum_{k=0}^{N-1} \mathcal{C}(x(k),u(k))$ 
taking into account the different communication delays. 


\begin{figure}
	\centering
	\includegraphics[width=0.8\columnwidth]{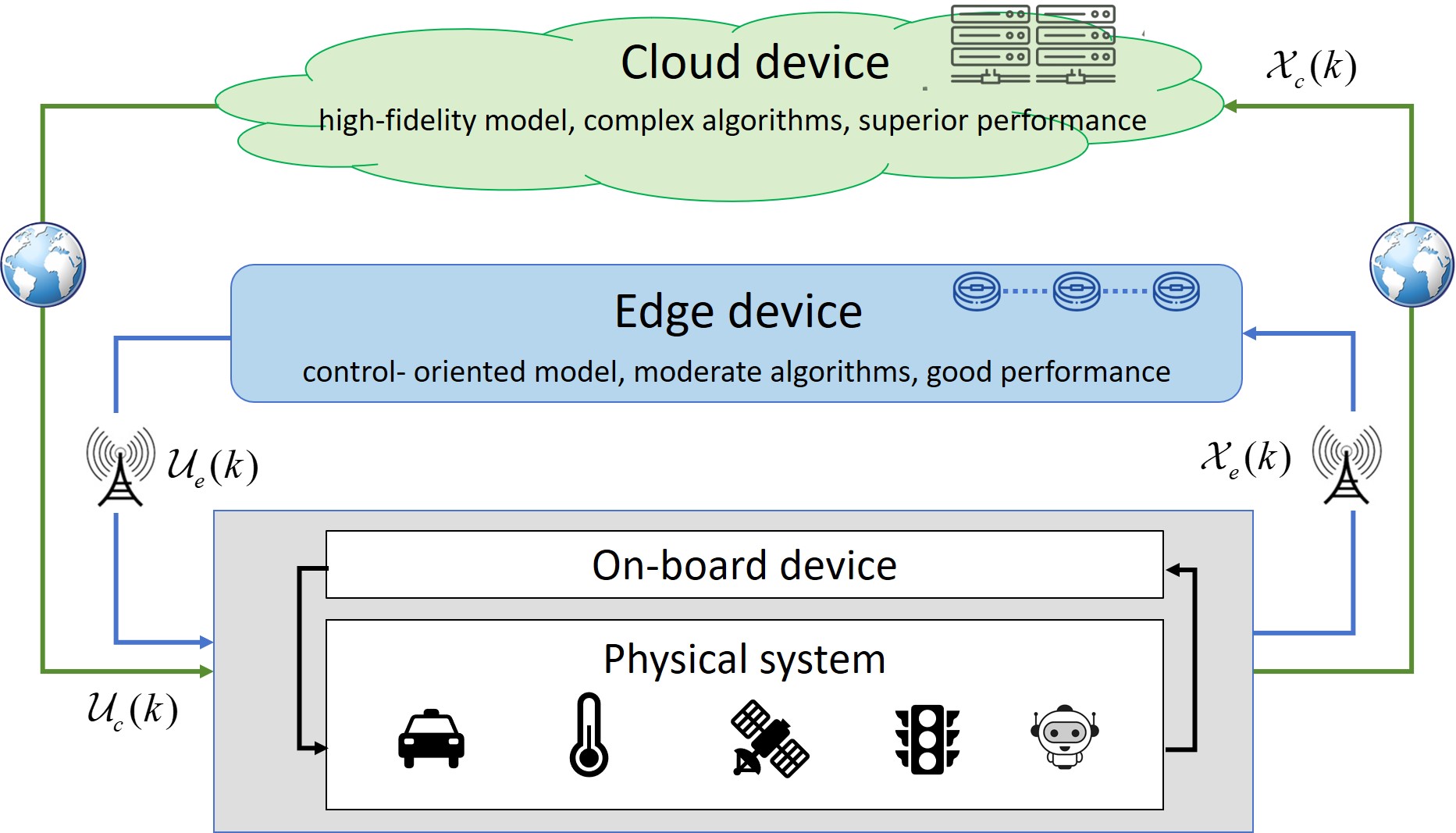}
	\caption{Computing Continuum}
	\label{framework} 
\end{figure}

\section{Proposed strategy}\label{proposed strategy}

\subsection{Network compensator}

Consider the cloud device.  
Due to the non-ideal uplink channel, the most recent packet received by the cloud device at time $k$ is $\mathcal{X}_c(k-a_c(k))$, that includes the state $x(k-a_c(k))$ and buffer $B(k-a_c(k))$. 
Due to the non-ideal downlink channel, the packet $\mathcal{U}_c(k)$ computed by the cloud at time instant $k$ will arrive at the plant with a delay $d_c(k)$.


The proposed strategy can be formalized as follows. First, we estimate the current state as
\begin{equation}
\hat{x}(k) = \hat{x}(a_c(k)\,|\,x(k\!-\!a_c(k)),B(k\!-\!a_c(k))) 
\end{equation}
where $\hat{x}(i|x,U)$ is the nominal predicted state $i$ steps ahead starting from state $x$ with input sequence $U=(u_0, \dots, u_{N-1})$. Formally, it can be computed recursively 
\begin{equation}\label{eq:nominal}
\hat{x}(i\!+\!1|x,U) = f(\hat{x}(i|x,U), u_i, 0)
\end{equation}
starting from $\hat{x}(0|x,U)=x$ for $i\leq N$. Essentially, we estimate the current state based on the nominal nonlinear model starting from state $x(k\!-\!a_c(k))$ and with inputs extracted from $B(k\!-\!a_c(k))$.
Second, based on the estimated state $\hat{x}(k)$, we predict the state $D_c$ step ahead, where $D_c$ is tuning parameters. Formally
\begin{equation}\label{ac+dc}
\hat{x}(k\!+\!D_c) = \hat{x}(D_c\,|\,\hat{x}(k),B_f(k\!-\!a_c(k))) 
\end{equation}
where $B_f(k\!-\!a_c(k))$ is buffer $B(k\!-\!a_c(k))$ without the first $a_c(k)$ elements, that are those used in the estimation of $\hat{x}(k)$. Note that we must have $D_c+a_c(k)\leq N$.

It is possible to show that
\begin{equation}
\hat{x}(k\!+\!D_c) = \hat{x}(D_c\!+\!a_c(k)\,|\,x(k\!-\!a_c(k)),B(k\!-\!a_c(k))) 
\end{equation}
Then, $\hat{x}(k\!+\!D_c)$ is regarded as the initial state of the given cloud controller, obtaining
\begin{equation}\nonumber
\begin{aligned}
&{U_c}(k + {D_c})\\
&~~~~= \{u_0^c(k + {D_c}), \ldots ,u_{N - 1}^c(k + {D_c})\} = {g_c}(\hat x(k + {D_c}))
\end{aligned}
\end{equation}
Finally, the control sequence is packetized as
\begin{equation}
\mathcal{U}_c(k)=U_c(k+D_c)
\end{equation}
and transmitted to the plant.

The same compensation scheme is also applied at the edge device. Based on the latest received packet $\mathcal {X}_e (k-a_e (k))=\{x(k-a_e(k)),B(k-a_e(k))\}$, the state is predicted $D_e$ steps ahead based on the nonlinear model as
\begin{equation}\label{ae+de}
\hat{x}(k+D_e) = \hat{x}(D_e\!+\!a_e(k)\,|\,x(k\!-\!a_e(k)),B(k\!-\!a_e(k))) 
\end{equation}
Correspondingly, the control sequence is generated as
\begin{equation}\nonumber
\begin{aligned}
&{U_e}(k + {D_e})\\
&~~~~= \{u_0^e(k + {D_e}), \ldots ,u_{N - 1}^e(k + {D_e})\} = {g_e}(\hat x(k + {D_e}))
\end{aligned}
\end{equation}
Finally, the control sequence is packetized as
\begin{equation}
\mathcal{U}_e(k)=U_e(k+D_e)
\end{equation}
and transmitted to the plant.

The numbers of steps $D_c$ and $D_e$ are a critical design parameter. On one hand, if $d(\mathcal{U}_c(k)) \leq D_c$, the control sequence $U_c(k\!+\!D_c)$ to be used at time instant $k\!+\!D_c$ is arrived and can be applied. Otherwise, if $d_c(\mathcal{U}(k)) > D_c$, the sequence $U_c(k\!+\!D_c)$ is not available at time instant $k\!+\!D_c$  and cannot be used. 
So, a larger $D_c$ increases the robustness to delays. On the other hand, the prediction accuracy deteriorates due to the possible disturbances and uncertainty on the real sequence of applied inputs. So, a larger $D_c$ causes a loss of performances due to the open-loop inaccurate prediction. The parameter $D_c$ needs to be chosen in order to trade-off the number of outdated packets and the prediction inaccuracy. 
A possible simple rule to choose the parameter $D_c$ is to enforce that the probability of outdated packets $Pr(d(\mathcal{U}_c(k)) > D_c)$ is below a certain desired threshold based on the distribution of the delays $d_c(k)$. The following straightforward lemma outlines this rule for the design of $D_c$ and $D_e$ for three relevant delay distributions. 
\begin{lemma}
Let $d(k)$ be a random delay process. For any $\rho > 0$ there exists a $D>0$ such that $Pr(d(k) > D) < \rho$. Moreover:
\begin{itemize}[topsep=-4pt]
\item If $d(k) \sim \mathrm{Exp}(\lambda)$, then $D \ge  \!-\! \frac{{\ln \rho }}{\lambda }$
\item If $d(k) \sim \mathcal{N}(\mu,\sigma)$, then $D \ge \mu  + \sigma {\Phi ^{ \!-\! 1}}\left( {1 - \rho } \right)$
\item If $d(k) \sim \mathrm{Log}\mathcal{N}(\mu,\sigma)$ then $D \ge \exp \left( {\mu  \!+\! \sigma {\Phi ^{ - 1}}\left( {1 \!-\! \rho } \right)} \right)$ 
\end{itemize}
where $\Phi(\cdot)$ is cumulative distribution function of $\mathcal{N}(0,1)$.
\end{lemma}

\subsection{Selection policy}


Preliminarily, we introduce the binary variable $\gamma_c(k)$ that is equal to $1$ if the control sequence $U_c(k)$ is available at the on-board unit at time instant $k$ and $0$ otherwise. Since $U_c(k)$ is included in the packet $\mathcal{U}_c(k-D_c)$, $\gamma_c(k)=1$ if $d(\mathcal{U}_c(k-D_c))\leq D_c$, namely if the communication delay of the packet is smaller than $D_c$. We define the binary variable $\gamma_e(k)$ related to $U_e(k)$ in the same way. 

We introduce an auxiliary temporary buffer of length $N$. The auxiliary buffer at time $k$ is defined as
\begin{equation}\label{B+}
B^+(k)=\{b_1(k-1),\dots, b_{N-1}(k-1), u^b_0(\hat{x}(k+N))\}
\end{equation}
with $\hat{x}(k+N)=\hat{x}(N|x(k),b_1(k-1),\dots, b_{N-1}(k-1))$. Essentially, the auxiliary buffer is derived from the buffer at the previous time instant, shifted of 1 step, and padded with a control input derived from the on-board control law. 

We introduce the auxiliary cost function
\begin{equation}\label{finite V_N}
\mathcal{V}_N(x,U) = \sum\nolimits_{i=0}^{N-1} \mathcal{C}(\hat{x}(i|x,U),u_i) + \mathcal{C}_N(\hat{x}(N|x,U))
\end{equation}
where $\mathcal{C}(x,u) $ is the stage cost and $\mathcal{C}_N(x)$ is the terminal cost.
We propose to update the buffer as
\begin{equation}\label{buffer_update}
\begin{aligned}
&B(k)= \arg \min _U \, \mathcal{V}_N(x(k),U)\quad \\
&{\rm{s}}{\rm{.t}}{\rm{.}}\quad U \in \left\{ {{\gamma _c}(k){U_c}(k),\;{\gamma _e}(k){U_e}(k),\;{U_b}(k),\;{B^ + }(k)} \right\}
\end{aligned}
\end{equation}
where, with a little abuse of notation, if $\gamma(k)=0$ then $\gamma(k)U(k)=\emptyset$. Essentially, the buffer stores the future candidate control sequence that achieves the smallest finite horizon cost among those available at the plant. 
Finally, the applied input is taken as the first element in the buffer
\begin{equation}\label{control action}
u_k = b_0(k)
\end{equation}
while the measured state and the buffer are packetized in $\mathcal{X}_c(k)$ and $\mathcal{X}_e(k)$ and transmitted to the cloud and the edge, respectively. 

A potential benefit of the proposed control framework comes from the auxiliary buffer $B^+(k)$.
In fact, when control packets from the cloud or the edge device are unavailable, $B^+(k)$ provides a possible control sequence for the plant obtained as an extension of the previous optimal control sequence.
At the same time, the proposed strategy is able to seemingly switch to the real-time feedback provided by the on-board controller. In fact, when the control sequence stored in the buffer is outdated and performances start deteriorating, the on-board controller is activated and a minimum level of control performances is guaranteed. 
In this way, the proposed framework is able to organically manage the available computing devices to optimize the control performance taking into account the communication delay. Note that the proposed framework can be immediately generalized to the case of any arbitrary number of controllers.

\subsection{Stability analysis}

The applied input is selected from three different controllers and the switching policy is affected by the stochastic delays. 
Moreover, prediction errors jeopardize guarantees on the remote controller. For these reasons, convergence might be not guaranteed.
In this subsection, we show how to guarantee the stability under the proposed strategy. 
The following technical assumptions are made.

\begin{assumption}\label{disturbance bounded}
	The disturbance $w$ is bounded within a compact set $W \subset {\mathbb{R}^n}$, i.e. $w \in W \subset {\mathbb{R}^n}$.
\end{assumption}

\begin{assumption}\label{Lipschitz of f}
	Model $f(x,u,0)$ is Lipschitz continuous with 
	$||f\left( {x_i,{u_i},0} \right) - f\left( {x_j,{u_j},0} \right)|| \le {\cal L}_{f,x}||x_i - x_j|| + {\cal L}_{f,u}||{u_i} - {u_j}||, \forall x_i,x_j\in \mathbb{R}^n, u_i,u_j \in \mathbb{R}^m$, with Lipschitz constants ${\cal L}_{f,x} \ge 0,{\cal L}_{f,u} \ge 0$.
	Moreover, $||f\left( {{x},{u},{w}} \right) - f\left( {{x},{u},0} \right)|| \le \varepsilon, \forall x\in \mathbb{R}^n, u \in \mathbb{R}^m, w \in W$.
\end{assumption}

\begin{assumption}\label{Lipschitz of cost}
	Stage cost $\mathcal{C} (x,u)$ is Lipschitz continuous in $x$, that is $||{{\cal C}}\left( {x_p,u_p} \right) - {{\cal C}}\left( {x_q ,u_q} \right)|| \le {{\cal L}_{\cal C}}||x_p - x_q||$, with Lipschitz constant ${\cal L}_{\cal C}\ge 0$.	
\end{assumption}

\begin{assumption}\label{terminal assumption}
	The terminal cost $\mathcal{C}_N(x)$ in (\ref{finite V_N}) satisfies
	\begin{equation}\nonumber
	\begin{aligned}
	{{\cal C}_N} \left( x \right) \; \ge {{\cal C}}\left( {x,u_0^b(x)} \right) + {{\cal C}_N}\left( {f\left( {x,u_0^b(x)} \right)} \right), \forall x\in\mathbb{R}^n
	\end{aligned}
	\end{equation}
\end{assumption}

Assumptions \ref{disturbance bounded} and \ref{Lipschitz of f} are typical assumptions on the system dynamics. Assumptions \ref{Lipschitz of cost} and \ref{terminal assumption} prescribe how to select the on-board controller and the terminal cost. Essentially, it is required that the on-board controller robustly stabilizes the system and the terminal cost in the selection policy is a suitable Lyapunov-like function. Similar assumptions are used [\cite{chisci2001systems}][\cite{rawlings2017model}].

\begin{theorem}\label{theorem 1}
Consider system (\ref{nonlinear model}) under proposed control selection (\ref{buffer_update}). Suppose that Assumptions \ref{disturbance bounded}, \ref{Lipschitz of f}, \ref{Lipschitz of cost}, \ref{terminal assumption} hold, then the closed-loop system is Input-to-State stable (ISS), i.e., there exist a class $\mathcal K \mathcal L$ function $\beta$ and a class $\mathcal K $ function $\sigma$ such that
\begin{equation}
||x\left( k \right)|| \le \beta \left( {||x\left( {{k_0}} \right)||,k - {k_0}} \right) + \sigma \left( \varepsilon  \right).
\end{equation}	
where $k_0$ is the initial step and $\varepsilon$ is the disturbance bound. 
\end{theorem}
\begin{proof}	
As stated in \cite{sontag2008input}, \textcolor{black}{a system is ISS if it admits an ISS-type Lyapunov function,} i.e., if there exists a positive definite function $ \mathcal{V}_N(x(k)),~\forall k\ge 0$ such that  
\begin{align}\color{red}
&{\alpha _1}\left( {x\left( k \right)} \right) \le  \mathcal{V}_N\left( {x\left( k \right)} \right) \le {\alpha _2}\left( {x\left( k \right)} \right)\\ 
& \mathcal{V}_N\left( {x\left( {k + 1} \right)} \right) -  \mathcal{V}_N\left( {x\left( k \right)} \right) \le  - {\alpha _3}\left( {x\left( k \right)} \right) + \eta \left( \varepsilon  \right)
\end{align}	
where $\alpha_i,~i=1,2,3$ are class $\mathcal{K}_\infty$ functions, $\eta$ belongs to class $\mathcal{K}$ function. 
In what follows, we show $ \mathcal{V}_N\left( {x\left( k \right),B(k)} \right)$ is an ISS-type Lyapunov function.

From (\ref{finite V_N}) we know 
\begin{align}
\mathcal C\left( {\hat x\left( {0,U} \right),{u_0}} \right) \le \mathcal{V}_N\left( {x\left( k \right),B(k)} \right) \le \mathcal V_N\left( {x\left( k \right),U_b\left( k \right)} \right)
\end{align}
and $\mathcal C(x(k),u(k))$, $\mathcal V_N\left( {x\left( k \right),U} \right)$ are a class $\mathcal{K}_\infty$ functions.

At step $k$, ${{\cal V}_N}(x(k),B(k))=\min \mathcal{V}_N(x(k),U)$ based on (\ref{buffer_update}).
For the next step $k+1$, it obtains from (\ref{B+}),
\begin{equation}
\begin{aligned}
&{{\cal V}_N}\left( {x(k + 1),{B^ + }\left( {k + 1} \right)} \right)\\
&~~= \sum\limits_{i = 0}^{N - 2} {\cal C} \left( {\hat x\left( {i|x\left( {k + 1} \right),{B^ + }\left( {k + 1} \right)} \right),{b_{i + 1}}} \right)\\
&~~+ {\cal C}\left( {\hat x\left( {N - 1|x\left( {k + 1} \right),{B^ + }\left( {k + 1} \right)} \right),u_0^b(\hat x(k + N))} \right)\\
&~~+ {{\cal C}_N}\left( {\hat x\left( {N|x\left( {k + 1} \right),{B^ + }\left( {k + 1} \right)} \right)} \right),
\end{aligned}
\end{equation}
	following
	\begin{equation}\label{V(k+1)-V*(k)}
	\begin{aligned}
	&{{\cal V}_N}\left( {x(k + 1),{B^ + }\left( {k + 1} \right)} \right) - {{\cal V}_N}\left( {x(k),B(k)} \right)\\
	&~~= \sum\limits_{i = 0}^{N - 2} {{\cal C}\left( {\hat x\left( {i|x\left( {k + 1} \right),{B^ + }\left( {k + 1} \right)} \right),{b_{i + 1}}} \right)} \\
	&~~- \sum\limits_{i = 1}^{N - 1} {{\cal C}\left( {\hat x\left( {i + 1|x\left( k \right),B\left( k \right)} \right),{b_{i + 1}}} \right)} \\
	&~~- {\cal C}\left( {x\left( k \right),{b_0}} \right) - {{\cal C}_N}\left( {\hat x\left( {N|x\left( k \right),B\left( k \right)} \right)} \right)\\
	&~~+ {\cal C}\left( {\hat x\left( {N - 1|x\left( {k + 1} \right),B^+(k+1)),u_0^b(\hat x(k + N)} \right)} \right)\\
	&~~+ {{\cal C}_N}\left( {\hat x\left( {N|x\left( {k + 1} \right),{B^ + }\left( {k + 1} \right)} \right)} \right)
	\end{aligned}
	\end{equation}
	According to the Assumptions \ref{disturbance bounded}, \ref{Lipschitz of f}, it derives
	\begin{equation}
	\begin{aligned}\label{f-f}
	||f\left( {\hat x\left( {i|k + 1} \right),{b_i},0} \right) - f\left( {\hat x\left( {i + 1|k} \right),{b_i},0} \right)||\\
	\color{black} \le {\cal L}_{f,x}^{N-2}|| x\left( {k + 1} \right) - \hat x\left( {1|k} \right)|| 
	\end{aligned}
	\end{equation}
	
	Based on Assumption \ref{Lipschitz of cost} and (\ref{f-f}), we rewrite (\ref{V(k+1)-V*(k)})  as
	\begin{equation}
	\begin{aligned}
	&{{\cal V}_N}\left( {x(k + 1),{B^ + }\left( {k + 1} \right)} \right) - {{\cal V}_N}\left( {x(k),B(k)} \right)\\
	&~~\le \sum\limits_{i = 0}^{N - 2} {{\cal L}_{\cal C} {\color{black}{\cal L}_{f,x}^i} } {||x\left( {k + 1} \right) - \hat x\left( {1|k} \right)||} \\
	&~~- {\cal C}\left( {x\left( k \right),{b_0}} \right)- {{\cal C}_N}\left( {\hat x\left( {N|x\left( k \right),B\left( k \right)} \right)} \right)\\
	&~~+ {\cal C}\left( {\hat x\left( {N - 1|x\left( {k + 1} \right),u_0^b(\hat x(k + N))} \right)} \right)\\
	&~~+ {{\cal C}_N}\left( {\hat x\left( {N|x\left( {k + 1} \right),{B^ + }\left( {k + 1} \right)} \right)} \right)
	\end{aligned}
	\end{equation}
	By Assumption \ref{Lipschitz of f},$||x\left( {k + 1} \right) - \hat x\left( {1|k} \right)||\le \varepsilon$, thus  
	\begin{equation*}
	\sum\limits_{i = 0}^{N - 2} {{\cal L}_{\cal C} {\color{black}{\cal L}_{f,x}^i} } ||x\left( {k + 1} \right) - \hat x\left( {1|k} \right)||\le \eta \left( \varepsilon  \right)
	\end{equation*} 
    for some class $\mathcal{K}$ function $\eta$.
    From Assumption \ref{terminal assumption}
	\begin{equation}
	\begin{aligned}
	{{\cal V}_N}\left( {x(k + 1),{B^ + }\left( {k + 1} \right)} \right)& - {{\cal V}_N}\left( {x(k),B(k)} \right)\\
	&\le  - {\cal C}\left( {x\left( k \right),{b_0}} \right)+\eta \left( \varepsilon  \right)
	\end{aligned}
	\end{equation}
	Noticing that
	\begin{equation}
	\begin{aligned}
	&{{\cal V}_N}\left( {x(k + 1),B\left( {k + 1} \right)} \right) 
	= \min {{\cal V}_N}(x(k + 1),U) \\
	&~~~~~~~~~~~~~~~~~~~~~~~~~\le {{\cal V}_N}\left( {x(k + 1),{B^ + }\left( {k + 1} \right)} \right)
	\end{aligned}
	\end{equation}
	it follows 
	\begin{equation}\label{VN(k+1)-VN(k)}
	\begin{aligned}
	{{\cal V}_N}\left( {x(k + 1),B\left( {k + 1} \right)} \right) &- {{\cal V}_N}\left( {x(k),B(k)} \right)\\
	&\le  - {\cal C}\left( {x\left( k \right),{b_0}} \right)+\eta \left( \varepsilon  \right)
	\end{aligned}
	\end{equation}
	The proof is completed.	
\end{proof}
Theorem \ref{theorem 1} indicates that the state $x(k)$ will converge to a bounded neighborhood of the origin constrained by the disturbance amplitude. If disturbance vanishes, the state $x(k)$  will convergence to the origin.
%
%
Besides, it is worth mentioning that no assumptions are made on the cloud and edge controllers.
In particular, it is not required that the cost is non-increasing under the edge or the cloud controller at any time instant. The use of a robust but simple on-board controller is enough to guarantee the stability of the overall scheme.
Remarkably, the assumptions pose conditions on the on-board controller and the cost taking into account the effects of model uncertainties and external disturbances, while no conditions are made on the communications.
The proposed framework is able to obtain safe evolution thanks to the on-board controller in any conditions while it can achieve higher performances in good channel conditions by exploiting the edge and/or the cloud.

\section{Simulations}\label{test}

We consider an industrial warehouse environment where a 4-wheel mobile robot is employed to move objects.
The well-known kinematic bicycle model is used to model the robot. At discrete-time, the model is
\begin{equation}\label{nonlinear discrete nodel}
x(k+1)= f(x(k),u(k)), 
\end{equation}
with $x = {\left[ {{p_x},{p_y},\varphi ,v} \right]^T}$ and $u = {\left[ {\beta ,a} \right]^T}$, defined as
\begin{align}
&p_x(k\!+\!1) = p_x(k)  + T v(k) \cos(\varphi(k) + \beta(k) ) \\
&p_y(k\!+\!1) = p_y(k)  + T v(k) \sin(\varphi(k) + \beta(k) ) \\
&\varphi(k\!+\!1) = \varphi(k) + T\frac{v(k)}{L_r} \sin(\beta(k)) \\
&v(k\!+\!1) = v(k) + T a(k)
\end{align}
with
\begin{align}
\beta(k) = \tan^{-1} \left( \frac{L_r}{L_r + L_f} \tan(\delta(k)) \right)
\end{align}
where $p=[p_x,p_y]$ is the position, $v$ is the speed, $a$ is the acceleration, $\varphi$ is the heading angle, $\beta$ is the slip angle, and $\delta$ is the steering angle.
See [\cite{ge2021numerically}].

The robot is required to reach the desired position $p^*$ and a fixed obstacle $\mathcal{O}$ is present. 
The objective is to accomplish the navigation task while minimizing the cost with
\begin{equation}
{\cal C}(p,u) = {\left\| {p - {p^*}} \right\|_Q} + {\rm{ }}{\left\| u \right\|_R} + c(p,{{\cal O}})
\end{equation}
where $c(p,\mathcal{O})$ is a penalty term to avoid the obstacle, i.e. $c(p,\mathcal{O}) = M  \gg 0 $ if $p \in \mathcal{O}$ and $0$ otherwise.

The cloud controller implements a MPC based the nominal nonlinear model. The horizon is $N$ and the stage cost is 
\begin{equation}
\mathcal{C}_\text{cloud}(p,u) = \mathcal{C}(p,u)
\end{equation}
Due to the nonlinearity of the model and nonlinearity of the cost function, this results into a complex nonlinear programming problem.

The edge controller implements a simplified linear MPC. In particular, since the penalty term in the cost might be intractable with the edge computational capabilities, the stage cost is set as
\begin{equation*}
\mathcal{C}_\text{edge}(p,u) ={\left\| {p - {p^*}} \right\|_Q} + {\rm{ }}{\left\| u \right\|_R} + c_\text{edge}(p,\mathcal{O}) 
\end{equation*} 
with
\begin{equation*}
c_{\text{edge}}(p,\mathcal{O}) = c e^{-\kappa({\left\|p-\tilde p \right\|} - r)}
\end{equation*}  
where $\tilde p$ is the center of the obstacle, $r$ is the radius, $\kappa$ is decay rate, and $c$ is the cost on the boundary of the obstacle. 
Moreover, the linearized model is used. 
The resulting problem is a convex optimization problem, which can be efficiently solved by the edge device.

The on-board controller implements the state feedback 
\begin{align*}
a(k) &= {K_p}\left( {v(k) - \sqrt {v_x^2(k) + v_y^2(k)} } \right)\\
\beta (k) &= {\tan ^{ - 1}}\left( {\frac{{{v_y}(k)}}{{{v_x}(k)}}} \right)- {\varphi(k)}
\end{align*}
where $v_x(k)$ and $v_y(k)$ are defined as
\begin{align*}
{v_x}(k) = p_x^* - {p_x}(k),~{v_y}(k) = p_y^* - {p_y}(k)
\end{align*}
The use of destination $p^*$ in the control law might lead to constraint violation. To avoid this, we consider a simple routine that gives an intermediate setpoint taken from a desired path enough in the future.
In particular, the desired path consists of arcs and line segments chosen to guarantee convergence and obstacle avoidance with a suitable safety margin. 

The sampling period is $T=0.01$. 
The cost matrices are $Q = diag(0.1,0.1)$, $R =diag(1.5, 1.5)$, the prediction horizon is $N=25$, and the penalty on the obstacle is $M=1000$.
For the edge controller, the cost on the boundary of the obstacle is $c=5000$, the decay rates is $\kappa=0.1$. For the on-board controller, we set $K_p=0.009$.
We set $D_e=2$ and $D_c=4$.

Based on the real Cloud-to-User Latency of AWS and Microsoft Azure [\cite{palumbo2021characterization}], log-normal distribution is used to model the delay in the cloud-plant link, while normal distribution is used to model the delay in the the edge-plant link. We set
\[\begin{array}{l}
{a_e},{d_e} \sim \mathcal{N}\left( {{\mu _e},\sigma _e^2} \right), \quad {a_c},{d_c} \sim \mathrm{Log} \mathcal{N}\left( {{\mu _c},\sigma _c^2} \right)
\end{array}\]
where we vary the parameters $\mu_e,\sigma _e^2,\mu_c,\sigma_c^2$ in order to evaluate different scenarios. 
\textcolor{black}{In the following, instead of reporting the values of parameters, we report the resulting probability of outdated (i.e., lost) packets for the considered threshold $D_c$ and $D_e$. In particular, we use the pair $\{p_c,p_e\}=\{Pr(d_c(k) > D_c),Pr(d_e(k) > D_e)\}$ to represent the considered cases.}

\begin{figure}[t]
	\centering
	\includegraphics[trim={0.5cm 0cm 0.5cm 0.72cm},clip, width=0.9\columnwidth]{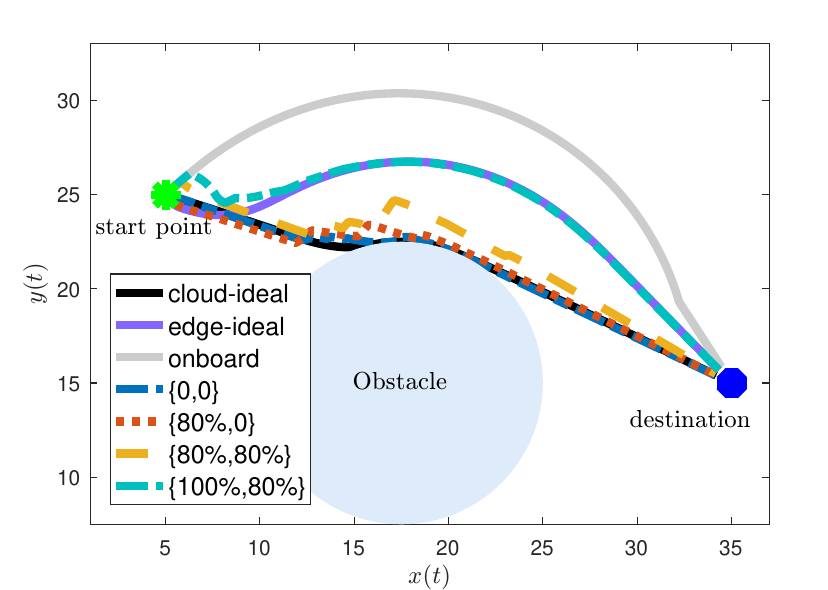} 
	\includegraphics[trim={0.5cm 0cm 0.5cm 0cm},clip, width=0.9\columnwidth]{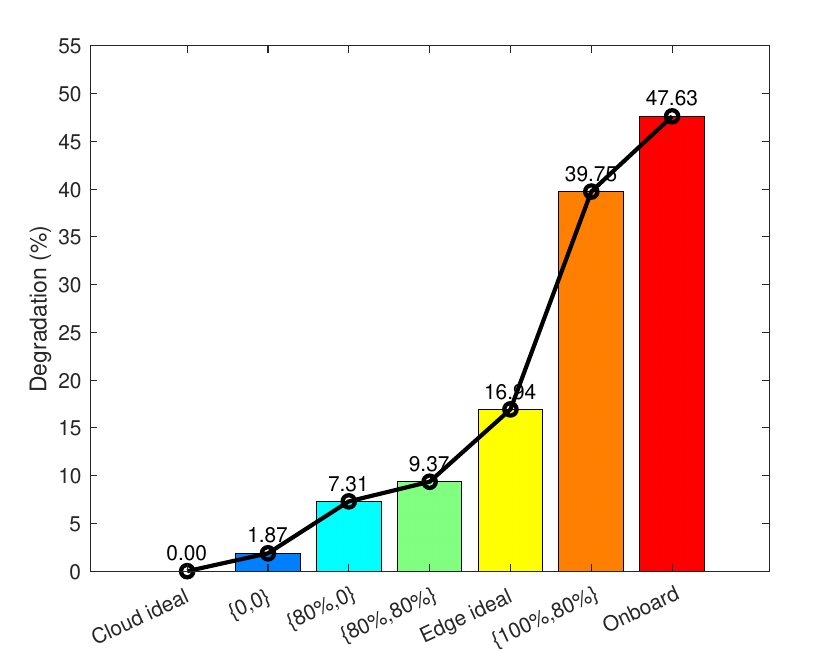}
	\caption{System trajectories (top) and cost degradation with 100 runs Monte Carlo test (bottom) varying channel conditions without disturbances. }
	\label{total_cost_pktloss_without_disturbances} 
\end{figure}

Figure \ref{total_cost_pktloss_without_disturbances} shows the results of the cloud, edge, on-board controllers considered individually for the nominal system without disturbance (i.e. $w(k)=0$) and of the proposed strategy with lossy channels. It is observed that the powerful computing capabilities of cloud devices and the sophisticated cloud controller are able to obtain ideal performances under ideal channel conditions. Remarkably, the trajectory generated by proposed control strategy under lossy channel is only slightly worse than the ideal behavior with ideal channels. In fact, even in the case of packet loss probability $\{p_c,p_e\}=\{80\%,0\}$, the trajectory follows quite closely the ideal trajectory. It is worth mentioning that in the case of $\{p_c,p_e\}=\{100\%,80\%\}$, that is equivalent to completely disable the communication between the cloud and the plant, the performance of the proposed strategy significantly deteriorates compared to the case of $\{p_c,p_e\}=\{80\%,80\%\}$. These results show the benefits of the three-tiered architecture over a simpler two-tiered architecture. In that case, relying on the edge controller instead of only on the on-board controller is particularly effective.   
Another interesting observation is the impact of the parameter $D_c$ on the control performances. As we can see in Figure  \ref{total_cost_pktloss_without_disturbances}, the degradation of the introducing the prediction horizon $D_c = 4$ when packet loss is not present is about $1.87\%$. This suggests that the open-loop prediction of $4$ steps only marginally affects the system performances while noticeably increasing the robustness against delays.

The comparison between $\{0,0\}$ and $\{80\%,0\}$ shows that even when the probability of cloud signal loss is 0.8 the system performance does not deteriorate significantly. 
The comparison between $\{80\%,80\%\}$ and $\{100\%,80\%\}$ shows that the performance improvement given by cloud is noticeable even when cloud signal is received infrequently. 
Motivated by this outcome, we study how often the cloud controller is applied with respect to the other controllers. We consider the case $\{80\%,80\%\}$. 
The applied input is selected 82\% of time from the buffer, 10\% from the cloud, 6\% from the edge, and 2\% from the on-board controller. Overall, a sequence generated from the cloud (possibly extracted from the buffer) is applied 83\% of time. 
This suggests that the cloud controller is particularly useful even if the packet is often lost.


Figure \ref{total_cost_pktloss_with_disturbances} shows the results in the case of disturbance following the uniform distribution  $w(k)\sim \mathcal{U}(-b,b)$, where $b_x=b_y=0.5$, $b_\varphi=b_v=0.1$. 
Due to the presence of disturbances, the performance of the proposed control strategy slightly deteriorates. However, it still noticeably outperforms the single-tier and double-tier controllers with the edge device.
Note that, in the presence of disturbances, even with ideal communications, it is not possible to guarantee always obstacle avoidance. A possible solution is to use robust MPCs such as Tube MPC.

\begin{figure}[tb]
	\centering
	\includegraphics[trim={0.25cm 0.25cm 0.5cm 0.8cm},clip, width=0.9\columnwidth]{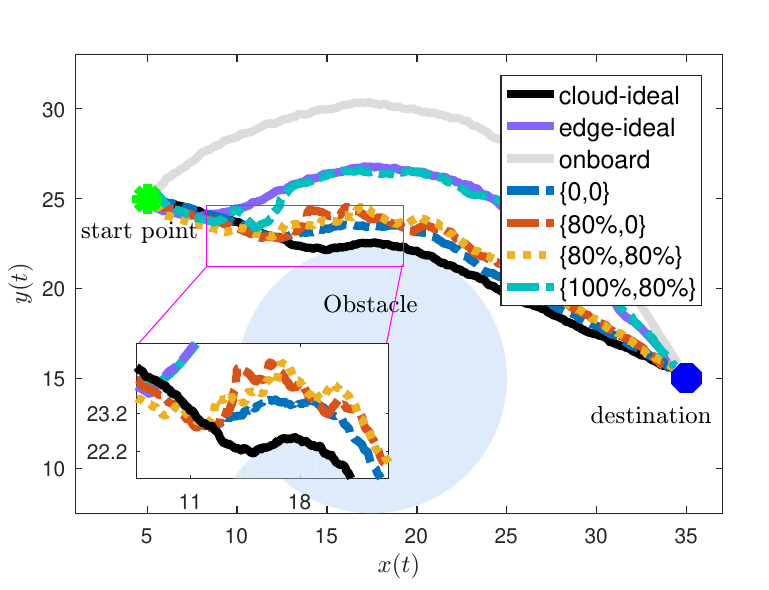}
    \caption{System trajectories with disturbance.}
	\label{total_cost_pktloss_with_disturbances} 
\end{figure}

\section{Conclusions}
In this work, we introduce a multi-tiered controller and a performance-oriented control selection policy to seamlessly exploit multiple computing devices over the edge-cloud continuum. We theoretically investigate the stability of the system and we show the benefits of the proposed setup. Future challenges are the design of multi-tiered constrained controllers and the study of the fundamental trade-offs of the edge-cloud architecture for control.

\bibliography{refs}           

\end{document}